\documentclass[letterpaper, 10 pt, conference]{ieeeconf}

\IEEEoverridecommandlockouts                              
\usepackage[english]{babel}
\usepackage{amsmath,amsfonts,amssymb, amsthm}
\usepackage{graphicx}
\usepackage[colorlinks=true, allcolors=blue]{hyperref}
\usepackage{parskip}
\usepackage{tabularx}
\usepackage{bm}
\usepackage{hhline}
\usepackage{multirow}

\usepackage{etoolbox}
\patchcmd{\thebibliography} 
  {\settowidth}
  {\setlength{\parsep}{0pt}\setlength{\itemsep}{0pt plus 0.1pt}\settowidth}
  {}{}

\newtheorem{theorem}{Theorem}

\newtheorem*{remark}{Remark}

\renewcommand{\t}{^{\mbox{\tiny\sf T}}} 
\newcommand{\blkdiag}{\text{blkdiag}}
\newcommand{\R}{\mathbb{R}}
\renewcommand{\P}{\mathbb{P}}
\newcommand{\E}{\mathbb{E}}
\newcommand{\tr}{\text{tr}}
\renewcommand{\det}{\text{det}}

\allowdisplaybreaks

\renewcommand{\baselinestretch}{0.979}

\title{\LARGE \bf Discrete-time Optimal Covariance Steering via\\ Semidefinite Programming}

\author{George Rapakoulias$^{1}$ and Panagiotis Tsiotras$^{2}$
\thanks{$^{1}$ Ph.D. student, School of Aerospace Engineering, Georgia Institute of Technology, Atlanta, GA, 30332, USA. Email:
 {\tt\small grap@gatech.edu}}%
\thanks{$^{2}$ David and Andrew Lewis Chair and Professor, School of Aerospace Engineering, and Institute for Robotics and Intelligent Machines, Georgia Institute of Technology, Atlanta, GA, 30332, USA. Email:
 {\tt\small tsiotras@gatech.edu}}%
}

\begin{document}

\maketitle

\begin{abstract}
This paper addresses the optimal covariance steering problem for stochastic discrete-time linear systems subject to probabilistic state and control constraints. 
A method is presented to efficiently attain the exact solution of the problem based on a lossless convex relaxation of the original non-linear program using semidefinite programming. 
Both the constrained and the unconstrained versions of the problem with either equality or inequality terminal covariance boundary conditions are addressed. We first prove that the proposed relaxation is lossless for all of the above cases.
Numerical examples are then provided to illustrate the proposed method. 
Finally, a comparative study is performed on systems of various sizes and steering horizons to illustrate the advantages of the proposed method in terms of computational resources compared to the state of the art. 
\end{abstract}

\section{Introduction}

The Covariance Control (CC) problem for linear systems was initially posed by A. Hotz and R. Skelton in \cite{hotz1987covariance}. 
It was studied in an infinite horizon setting for both continuous and discrete-time systems and the authors provided a parametrization for all linear state feedback controllers that achieve a specified system covariance. 
Later, the authors in \cite{grigoriadis1997minimum} provided analytical solutions for the minimum effort controller that achieves a specified steady-state system covariance in the same setting.

Its finite horizon counterpart, the Covariance Steering (CS) problem, gained attention only recently.
Although similar ideas can be traced back in the  Stochastic Model Predictive Control literature~\cite{primbs2009stochastic,farina2013probabilistic}, in the sense that these methods also try to address constraints in the system covariance, they achieve this objective by using conservative approximations or by solving computationally demanding non-linear programs.
Covariance Steering theory, on the other hand, 
offers a more direct approach, often providing tractable algorithms for the solution in real time.

The first formal treatment of the CS problem was provided in \cite{chen2015optimal,chen2015optimal1} for continuous-time systems, by studying the minimum-effort finite horizon covariance steering problem in continuous time. 
Later, in \cite{bakolas2018finite} the author provided a numerical approach for solving the discrete version of the problem with a relaxed terminal covariance boundary condition using semidefinite programming.
In \cite{okamoto2018optimal} the authors introduced a constrained version of the original problem where the state and control vectors are required to remain within specified bounds in a probabilistic sense. 
Finally, its connections to Stochastic Model Predictive control were cemented in \cite{okamoto2019stochastic}.

The recently developed covariance steering theory has  been applied to a variety of problems ranging from path planning for linear systems under uncertainty \cite{okamoto2019optimal}, control of linear systems with multiplicative noise \cite{liu2022optimal_mult}, distributed robot control \cite{saravanos2021distributed}, as well as for control of non-linear \cite{ridderhof2019nonlinear, saravanos2022distributed} and non-Gaussian \cite{sivaramakrishnan2022distribution, renganathan2022distributionally} systems. 
In our previous work \cite{liu2022optimal}, we presented a new method of solving the optimal covariance steering problem in discrete time based on an exact convex relaxation of the original non-linear programming formulation of the problem. 
At the same time, but independently, the authors of \cite{balci2022covariance} used the same relaxation to solve the optimal covariance steering problem with an inequality terminal boundary condition for a system with multiplicative noise.

The contributions of this paper are two-fold. First, we extend our previous results and prove that the proposed lossless convex relaxation presented in \cite{liu2022optimal} also holds under state and control chance constraints, as well as for the case of inequality terminal boundary covariance constraint.
The motivation for this extension is straightforward; many practical applications of covariance steering theory require probabilistic constraints to characterize the feasible part of state space or limit the control effort applied to the system. 
Furthermore, the inequality terminal covariance boundary condition might better reflect the desire to limit the uncertainty of the state, rather than driving it to an exact value. 
In this paper, we establish that the proposed method can handle all variants of the optimal covariance steering problem for linear systems encountered in the literature. 
Finally, we show that the proposed method outperforms other approaches for solving the CS problem, such as \cite{bakolas2018finite} and \cite{okamoto2019stochastic}, by over an order of magnitude in terms of run-times, while also having much better scaling characteristics with respect to the steering horizon and model size. 

\section{Problem Statement}
Let a stochastic, discrete, time-varying system be described by the state space model
\begin{equation}
    x_{k+1} = A_k x_k + B_k u_k + D_k w_k,
\end{equation}
where $k=0,1,\ldots,N-1$ denotes the time step, $A_k \in \R^{n \times n}$ is the system matrix, $B_k \in \R^{n \times p}$ is the input matrix and $D_k \in \R^{n \times q}$ is the disturbance matrix. The system's state, input, and stochastic disturbance are denoted by $x_k, \; u_k$ and $w_k$, respectively. 
The first two statistical moments of the state vector are denoted by $\mu_k = \E[x_k] \in \R^n$ and $ \Sigma_k = \E [ (x_k - \mu_k) (x_k - \mu_k)\t ] \in \R^{n \times n}$.
We assume that the process noise $w_k$ has zero mean and unitary covariance.
The discrete-time finite horizon optimal covariance steering problem can be expressed as the following optimization problem: 
\begin{subequations} \label{reference_problem}
\begin{align}
& \min_{x_k, u_k} \quad J =\E [\sum_{k = 0}^{N-1} { x_k\t Q_k x_k + u_k \t R_k u_k}], \label{ref:cost}
\end{align}
such that, for all $k = 0, 1, \dots, N-1$,
\begin{align}
& x_{k+1} = A_k x_k + B_k u_k + D_k w_k, \label{ref:dyn} \\
& x_0 \sim \mathcal{N}(\mu_i, \Sigma_i), \label{ref:initial_distr} \\
& x_N \sim \mathcal{N}(\mu_f, \Sigma_f), \label{ref:final_distr} \\
& \P(x_k \in \mathcal{X}) \geq 1-\epsilon_1, \label{ref:chan_con_x}\\
& \P(u_k \in \mathcal{U}) \geq 1-\epsilon_2 . \label{ref:chan_con_u}
\end{align}
\end{subequations}
For the rest of this paper, we will assume that $R_k \succ 0, \; Q_k \succeq 0$ and that $A_k$ is invertible for all $k = 0,1,\dots, N-1$. 

The decision variables for problem \eqref{reference_problem} are stochastic random variables, rendering it hard to solve using numerical optimization methods. 
As shown in \cite{liu2022optimal}, in the absence of the chance constraints \eqref{ref:chan_con_x}, \eqref{ref:chan_con_u} this problem is solved optimally with a linear state feedback law of the form 
\begin{equation}\label{controller}
u_k = K_k(x_k-\mu_k) + v_k,
\end{equation}
where $K_k \in \R^{ p \times n}$ is a feedback gain that controls the covariance dynamics and $v_k \in \R^p$ is a feedforward term controlling the system mean. 
Using \eqref{controller}, the cost function can be written, alternatively, in terms of the first and second moments of the state as follows 
\begin{equation*}
    J = \sum_{k = 0}^{N-1} \tr (Q_k \Sigma_k) + \tr(R_k K_k \Sigma_k K_k \t) + \mu_k\t Q_k \mu_k + v_k\t R_k v_k.
\end{equation*}
If the initial distribution of the state is Gaussian and a linear feedback law as in \eqref{controller} is used, the state distribution remains Gaussian.
This allows us to write the constraints \eqref{ref:initial_distr} and \eqref{ref:final_distr} as 
\begin{equation*}
    \mu_0 = \mu_i, \quad \Sigma_0 = \Sigma_i, \quad \mu_N = \mu_f, \quad \Sigma_N = \Sigma_f.
\end{equation*}

In contrast to previous works such as \cite{bakolas2018finite} and \cite{okamoto2019optimal}, we choose to keep the intermediate states in the steering horizon as decision variables, handling them in terms of their first and second moments. 
To this end, we replace \eqref{ref:dyn} with the mean and covariance propagation equations
\begin{subequations}
    \begin{align}
        & \mu_{k+1} = A_k\mu_k + B_k v_k, \\
        & \Sigma_{k+1} = (A_k + B_k K_k) \Sigma_k (A_k+B_k K_k)\t + D_k D_k\t.
    \end{align}
\end{subequations}
Omitting the chance constraints \eqref{ref:chan_con_u} and \eqref{ref:chan_con_x} for the moment, the problem is recast as a standard non-linear program
\begin{subequations} \label{NLP}
\begin{align}
& \min_{\Sigma_k, K_k, \mu_k, v_k} J = \sum_{k = 0}^{N-1} \tr (Q \Sigma_k) + \tr(R K_k \Sigma_k K_k \t) \nonumber \\
& \qquad \qquad \qquad \qquad + \mu_k\t Q \mu_k + v_k\t R_k v_k,  \label{NLP:cost}
\end{align}
such that, for all $k = 0, 1, \dots, N-1$,
\begin{align}
& \Sigma_{k+1} = A_k \Sigma_k A_k\t + B_k K_k \Sigma_k A_k\t + A_k \Sigma_k K_k\t B_k\t  \nonumber \\
& \quad \qquad + B_k K_k \Sigma_k K_k \t B_k\t + D_k D_k\t, \label{NLP:cov_dyn} \\
&  \Sigma_0 = \Sigma_i,  \label{NLP:in_cov} \\
&  \Sigma_N = \Sigma_f, \label{NLP:final_cov} \\
& \mu_{k+1} = A_k \mu_k + B_k v_k, \label{NLP:mean_dyn} \\
& \mu_0 = \mu_i, \label{NLP:in_mean} \\
& \mu_N = \mu_f. \label{NLP:final_mean} 
\end{align}
\end{subequations}
In the following sections, we will convert this problem to an equivalent convex one. 

\section{Unconstrained Covariance Steering}

It is well established in the covariance steering literature that under no coupled mean-covariance constraints, problem \eqref{NLP} can be decoupled into the mean steering problem and the covariance steering problem \cite{bakolas2018finite, okamoto2018optimal}.
The solution to the mean steering is trivial, therefore, we focus solely on the covariance steering, which corresponds to the optimization problem
\begin{align}
\min_{\Sigma_k, K_k} \quad & J_\Sigma = \sum_{k = 0}^{N-1} {\tr\big(Q_k \Sigma_k \big) + \tr \big(R_k K_k \Sigma_k K_k\t \big)}, \label{NLPcov} \\ 
& \textrm{subject to \eqref{NLP:cov_dyn}-\eqref{NLP:final_cov}}, \nonumber 
\end{align}
Using the change of variables $U_k = K_k \Sigma_k$ and the convex relaxation proposed in \cite{liu2022optimal} one can transform Problem \eqref{NLPcov} into a linear semidefinite program 
\begin{subequations} \label{convex_eq}
\begin{align}
& \min_{\Sigma_k, U_k, Y_k} \quad J_\Sigma = \sum_{k = 0}^{N-1} {\tr \big(Q_k \Sigma_k \big) + \tr \big(R_k Y_k \big)} 
\end{align}
such that, for all $k = 0, 1, \dots, N-1$,
\begin{align}
& C_k \triangleq U_k \Sigma_k^{-1} U_k\t - Y_k \preceq 0, \label{convex_eq:relaxation} \\
& G_k \triangleq A_k \Sigma_k A_k\t + B_k U_k A_k\t + A_k U_k\t B_k\t + B_k Y_k B_k\t \nonumber \\
& \qquad + D_k D_k\t - \Sigma_{k+1} = 0,\label{convex_eq:cov_dyn} \\
& \Sigma_N - \Sigma_f = 0, \label{convex_eq:final_cov} 
\end{align}
\end{subequations}
where the constraint \eqref{convex_eq:relaxation} can be expressed as an LMI using the Schur complement as
\begin{equation*}
\begin{bmatrix} 
\Sigma_k & U_k\t \\
U_k & Y_k
\end{bmatrix} \succeq 0.
\end{equation*} 
\begin{theorem} \label{thm:lossless1}
The optimal solution to the relaxed problem \eqref{convex_eq} satisfies $C_k = 0$ for all $k = 0,1,\dots, N-1$ and therefore also optimally solves \eqref{NLPcov} \cite{liu2022optimal}.
\end{theorem}
\begin{remark}
A different approach that results in the same formulation is that of a randomized feedback control policy presented in \cite{balci2022exact}. 
Therein, the injected randomness on the control policy can be interpreted as a slack variable converting \eqref{convex_eq:relaxation} to equality. 
In \cite{balci2022exact} it is shown that for the soft-constrained version of the problem, the value of this slack variable is zero. 
In our work, we tackle directly the hard-constrained version, instead, with equality or inequality terminal covariance constraints as well as chance constraints. 
In this case, strong duality is not apparent and the technique of the proof of \cite{balci2022exact} is not directly applicable.
\end{remark}
Next, consider Problem \eqref{NLPcov} but with an inequality terminal covariance boundary condition instead, and its corresponding relaxed version, namely,
\begin{subequations} \label{convex_ineq}
\begin{align}
& \min_{\Sigma_k, U_k, Y_k} \quad J_\Sigma = \sum_{k = 0}^{N-1} {\tr \big(Q_k \Sigma_k \big) + \tr \big(R_k Y_k \big)},
\end{align}
such that, for all $k = 0, 1, \dots, N-1$,
\begin{align}
& C_k \triangleq U_k \Sigma_k^{-1} U_k\t - Y_k \preceq 0, \label{convex_ineq:relaxation} \\
& \Sigma_N - \Sigma_f \preceq 0, \label{convex_ineq:cov} \\
& G_k \triangleq A_k \Sigma_k A_k\t + B_k U_k A_k\t + A_k U_k\t B_k\t + B_k Y_k B_k\t \nonumber \\
& \qquad + D_k D_k\t - \Sigma_{k+1} = 0. \label{convex_ineq:cov_dyn}
\end{align}
\end{subequations}

\begin{theorem} \label{thm:lossless2}
Assuming that the exact covariance steering problem \eqref{NLPcov} is feasible, problem \eqref{convex_ineq} satisfies $C_k = 0$ for all $k = 0,1,\dots,N-1$ 
and therefore also optimally solves \eqref{NLPcov}  with an inequality terminal covariance boundary condition, instead of \eqref{NLP:final_cov}.
\end{theorem}
\begin{proof}
Using matrix Lagrange multipliers $M_k^{(1)}, \; M^{(2)}$, $ \Lambda_k$ for the constraints \eqref{convex_ineq:relaxation}, \eqref{convex_ineq:cov}, \eqref{convex_ineq:cov_dyn}, respectively, we define the Lagrangian function 
\begin{equation*}
\mathcal{L}( \cdot ) = J_{\Sigma} + \tr \big( M^{(2)} (\Sigma_N - \Sigma_f) \big) + \sum_{k = 0}^{N-1}{\tr \big( M_k^{(1)} C_k \big) + \tr \big( \Lambda_k G_k \big)}.
\end{equation*}
The relevant first-order optimality conditions are \cite{vandenberghe1996semidefinite}:
\begin{subequations}\label{opt_cond}
\begin{align}
& \frac{\partial \mathcal{L}}{\partial U_k} = 2 M_k^{(1)} U_k \Sigma_k^{-1} + 2 B_k\t \Lambda_k A_k = 0, \label{opt_cond_U} \\ 
& \frac{\partial \mathcal{L}}{\partial Y_k} = R_k - M_k^{(1)} + B_k\t \Lambda_k B_k = 0, \label{opt_cond_Y}\\ 
& \tr \big(M_k^{(1)} C_k \big) = 0, \label{opt_cond_comp_slack1} 
\end{align}
\end{subequations}
where $k = 0,1,\dots N-1$. 
Note that we can choose $\Lambda_k$ to be symmetric because of the symmetry of the constraint \eqref{convex_ineq:cov_dyn}, while $M_k^{(1)} $ and $ M^{(2)}$ are symmetric by definition. 
We will prove that the optimal solution to problem \eqref{convex_ineq} satisfies $C_k = 0$ for all $k = 0,1,\ldots,N-1$.
To this end, assume that $C_k$ has at least one nonzero eigenvalue for some $k$. 
Equation \eqref{opt_cond_comp_slack1} then yields that $M_k^{(1)}$ has to be singular \cite{liu2022optimal}. 
The optimality condition \eqref{opt_cond_U} can then be rewritten as $B_k\t \Lambda_k = -M_k^{(1)} U_k \Sigma_k^{-1} A_k^{-1}$. Substituting to \eqref{opt_cond_Y} yields
\begin{equation} \label{contradicion}
R_k = M_k^{(1)} \big( I_{p} + U_k \Sigma_k^{-1} A_k^{-1} B_k \big). 
\end{equation}
Calculating the determinants of both sides of \eqref{contradicion}, we obtain 
\begin{equation*}
\det(R_k) = \det(M_k^{(1)})\, \det\big( I_{p} + U_k \Sigma_k^{-1} A_k^{-1} B_k \big) = 0. 
\end{equation*}
This clearly contradicts the fact that $R_k \succ 0$. 
Therefore, at the optimal solution, the matrix $C_k$ has all its eigenvalues equal to zero. 
This, along with the fact that $C_k$ is symmetric, yields that $C_k = 0$ for all $k = 0,1,\dots,N-1$.
The final step to conclude the proof is to show that the KKT conditions \eqref{opt_cond} for the relaxed problem \eqref{convex_ineq} are sufficient for the optimal solution, or in other words, the duality gap for the relaxed problem is zero. 
We have already proved that strong duality holds for the exact covariance steering problem in \cite{liu2022optimal}.
Since the relaxed terminal boundary condition problem \eqref{convex_ineq} has a domain at least as big as the exact problem \eqref{convex_eq} and strong duality holds for the exact problem, from Slater's condition strong duality holds for the relaxed problem as well.  
\end{proof}

\section{Constrained Covariance Steering}

Many real-world applications require additional constraints of the form \eqref{ref:chan_con_u}, \eqref{ref:chan_con_x} to be imposed on the problem to reflect the physical limitations of the system or some other desired behavior. 
These may include constraints on the total control effort $u_k$ on each time step or physical limits on the state vector $x_k$. 
In this work, we assume polytopic state and control constraints of the form
\begin{subequations} \label{probabilistic_con}
    \begin{align}
        & \P(\alpha_x \t x_k \leq \beta_x) \geq 1-\epsilon_x, \\
        & \P(\alpha_u \t u_k \leq \beta_u) \geq 1-\epsilon_u,
    \end{align}
\end{subequations}
where $\alpha_x \in \R^{n}, \; \alpha_u \in \R^p,\; \beta_x, \beta_u \in \R $ and  $\epsilon_x, \epsilon_u \in [0, 0.5]$ reflects the violation probability of each constraint. 
To convert the probabilistic constraints~\eqref{probabilistic_con}  into deterministic constraints on the decision variables note that $\alpha_x\t x_k$ and $\alpha_u \t u_k$ are univariate random variables with first and second moments given by 
\begin{subequations} 
    \begin{align}
      & \hspace*{-3mm} \E(\alpha_x\t x_k) =  \alpha_x\t \mu_k, \\
        & \hspace*{-3mm}\E(\alpha_u \t u_k)  = \alpha_u \t v_k, \\
        & \hspace*{-3mm}\E ( \alpha_x\t (x_k -\mu_k) (x_k -\mu_k)\t \alpha_x)  =  \alpha_x\t \Sigma_k \alpha_x, \\
        &\hspace*{-3mm} \E(\alpha_u \t K_k (x_k-\mu) (x_k-\mu)\t K_k \t \alpha_u )  =  \alpha_u \t U_k \Sigma^{-1}_k U_k\t \alpha_u. \label{control_cov}
    \end{align}
\end{subequations}
To this end, according to \cite{okamoto2018optimal}, equations \eqref{probabilistic_con} are converted to 
\begin{subequations} \label{constraints}
    \begin{align}
        & \Phi^{-1}(1-\epsilon_x) \sqrt{ \alpha_x\t \Sigma_k \alpha_x} + \alpha_x\t \mu_k - \beta_x \leq 0, \label{chance_con} \\
        & \Phi^{-1}(1-\epsilon_u) \sqrt{ \alpha_u \t U_k \Sigma^{-1}_k U_k\t \alpha_u } + \alpha_u \t v_k - \beta_u \leq 0, \label{effort_st1}
    \end{align}
\end{subequations}
where $\Phi^{-1}(\cdot)$ is the inverse cumulative distribution function of the normal distribution. 
If the Gaussian assumption for the disturbances is dropped, then $\Phi^{-1}(\cdot)$ can be conservatively replaced using Cantelli's concentration inequality with $Q( 1 - \epsilon ) = \sqrt{{\epsilon}/{ (1 - \epsilon)}}$ \cite{renganathan2022distributionally}.

Using the same relaxation as before to handle the non-linear term $U_k \Sigma^{-1}_k U_k\t$, equation \eqref{effort_st1} is further relaxed to 
\begin{equation}    \label{effort_con}
\Phi^{-1}(1-\epsilon_u) \sqrt{ \alpha_u \t Y_k \alpha_u } + \alpha_u \t v_k - \beta_u \leq 0. 
\end{equation}
\begin{figure}
  \centering
  \includegraphics[width = .8\linewidth]{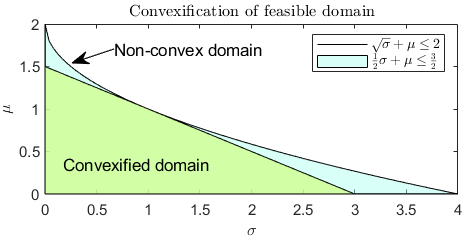}
  \caption{Example of a convexified domain for a 1-dimensional system}
  \label{fig: domain_convexification}
\end{figure}
Unfortunately, due to the square root on the decision variables $\Sigma_k$ and $Y_k$ neither of \eqref{chance_con}, \eqref{effort_con} are convex. One conservative option to overcome this issue is to linearize these constraints around some reasonable value of $\alpha_x\t \Sigma_k \alpha_x$ and $\alpha_u \t Y_k \alpha_u$, respectively, for a given problem. Because the square root is a concave function, the tangent line can serve as a linear global overestimator, yielding  
\begin{equation*}
    \sqrt{x} \leq \frac{1}{2 \sqrt{x_0}}x + \frac{\sqrt{x_0}}{2}, \quad \forall x,x_0 > 0.
\end{equation*}
The constraints in \eqref{constraints} can therefore be conservatively approximated as 
\begin{subequations} \label{constraints_lin}
    \begin{align}
         \Phi^{-1}(1-\epsilon_x) &\frac{1}{2 \sqrt{ \alpha_x\t \Sigma_r \alpha_x}} \alpha_x\t \Sigma_k \alpha_x + \alpha_x\t \mu_k  \nonumber\\
        & - \left( \beta_x - \Phi^{-1}(1-\epsilon_x) \frac{1}{2} \sqrt{ \alpha_x\t \Sigma_r \alpha_x} \right) \leq 0, \label{chance_con_lin} \\
         \Phi^{-1}(1-\epsilon_u) &\frac{1}{2 \sqrt{ \alpha_u \t Y_r \alpha_u }} \alpha_u \t Y_k \alpha_u + \alpha_u \t v_k \nonumber \\
        & - \left( \beta_u - \Phi^{-1}(1-\epsilon_u) \frac{1}{2} \sqrt{ \alpha_u \t Y_r \alpha_u} \right) \leq 0, \label{effor_con_lin}
    \end{align}
\end{subequations}
where $\Sigma_r, \; Y_r$ are some reference values. The linearized constraints now form a convex set, as illustrated in Figure \ref{fig: domain_convexification}. 
For notational simplicity, next, we consider the more general constraint form of 
\begin{subequations} \label{constraints_final}
    \begin{align}
        &  \ell \t \Sigma_k \ell + \alpha_x\t \mu_k - \beta_x \leq 0, \label{constraints_final_cov} \\ 
        &  e\t Y_k e + \alpha_u \t v_k - \beta_u \leq 0. 
    \end{align}
\end{subequations}
Given the additional constraints in \eqref{constraints_final} the fundamental question is whether the relaxation proposed in \eqref{convex_eq} remains lossless.
To this end, consider the constrained Covariance Steering problem
\begin{subequations} \label{convex_constrained}
\begin{align}
& \min_{\Sigma_k, U_k, Y_k, \mu_k, v_k} J, 
\end{align}
such that, for all $k = 0, 1, \dots, N-1$,
\begin{align}
& \mu_{k+1} = A_k \mu_k + B_k v_k, \\
& C_k(\Sigma_k, Y_k, U_k) \preceq 0,  \\
& G_k(\Sigma_{k+1}, \Sigma_k, Y_k, U_k) = 0, \\
& \ell \t \Sigma_k \ell + \alpha_x\t \mu_k - \beta_x \leq 0, \label{convex_constrained_cov} \\
& e\t Y_k e + \alpha_u \t v_k - \beta_u \leq 0.        \label{convex_constrained_effort} 
\end{align}
\end{subequations}
where $J$ is defined in \eqref{NLP:cost}. Note that an equality terminal covariance condition is implied, by excluding $\Sigma_N$ from the optimization variables and treating it as constant.
\begin{theorem} \label{lossless_2}
The optimal solution to the problem \eqref{convex_constrained} satisfies $C_k = 0$ for all $k = 0,1,\dots,N-1$. 
\end{theorem}

\begin{proof}
Define again the problem Lagrangian as 
\begin{equation*} 
\begin{split}
    \mathcal{L}_a( \cdot ) =  J & +  \sum_{k = 0}^{N-1} \tr \big( M_k\t C_k \big) + \tr \big( \Lambda_{1,k}\t G_k \big) \\
    & + \lambda_{1,k}\t (\mu_{k+1} - A_k \mu_k - B_k v_k) \\
    & + \lambda_{2,k} \big(\ell \t \Sigma_k \ell + \alpha_x\t \mu_k - \beta_x \big) \\ 
    & + \lambda_{3,k} \big(e\t Y_k e + \alpha_u \t v_k - \beta_u \big). 
\end{split}
\end{equation*}
The relevant first-order optimality conditions for this problem are
\begin{subequations}\label{opt_cond_constrained}
\begin{align}
%
& \frac{\partial \mathcal{L}_a}{\partial U_k} = 2 M_k U_k \Sigma_k^{-1} + 2 B_k\t \Lambda_k A_k = 0, \label{opt_cond_constrained_U}\\ 
&\frac{\partial \mathcal{L}_a}{\partial Y_k} = R_k - M_k + B_k\t \Lambda_k B_k  + \lambda_{3,k} e e \t = 0, \label{opt_cond_constrained_Y}\\ 
%
%
& \tr \big(M_k C_k \big) = 0. \label{opt_cond_constrained_comp_slack1}
\end{align}
\end{subequations}
Following the same steps as in the proof of Theorem \ref{thm:lossless1}, let $C_k$ have at least one nonzero eigenvalue. From \eqref{opt_cond_constrained_comp_slack1}, $M_k$ has to be singular. Solving for $B_k\t \Lambda_k$ in \eqref{opt_cond_constrained_U} and substituting in \eqref{opt_cond_constrained_Y} we get  
\begin{equation} \label{contradicion2}
R_k  + \lambda_{3,k} e e\t = M_k \big( I_{p} + U_k \Sigma_k^{-1} A_k^{-1} B_k \big). 
\end{equation}
Since $\lambda_{3,k} \geq 0$ by definition, and $e e\t \succeq 0$, it follows that $R_k  + \lambda_{2,k} e e\t \succ 0$. Therefore, taking again the determinant of both sides of \eqref{contradicion2} leads to a contradiction.
\end{proof}

\section{Numerical examples and run-time analysis}
To illustrate our method, we will first consider the problem of path planning for a quadrotor in a 2D plane. 
The lateral and longitudinal dynamics of the quadrotor will be modeled as a triple integrator, with state matrices
\begin{equation*}
    A = \begin{bmatrix} \ I_2 & \Delta T I_2 & 0_2 \\ 0_2 & I_2  & \Delta T I_2 \\ 0_2 & 0_2 & I_2 \end{bmatrix}, \quad B = \begin{bmatrix} 0_2 \\ 0_2 \\ \Delta T I_2  \end{bmatrix}, \quad  D = 0.1 I_6 ,
\end{equation*}
a time step of $\Delta T = 0.1$~sec, a horizon of $N=60$ and boundary conditions 
\begin{align*}
& \Sigma_i = I_6, \quad \Sigma_f = 0.1 I_6, \\
& \mu_i =  \begin{bmatrix} \ 20 & 0_{1 \times 5}  \end{bmatrix}\t, \quad  \mu_f = 0_{6 \times 1}.
\end{align*}
The feasible state space and control input space are characterized by bounding boxes expressed in the form of \eqref{constraints_lin} with parameters
\begin{align*}
    & \alpha_x = \Big\{ \left[\begin{array}{ccc} \pm 1 & 0 & 0_{1 \times 4} \end{array}\right]\t, \left[\begin{array}{ccc} 0 & \pm 1 & 0_{1 \times 4} \end{array}\right]\t \Big\}, \\
    & \beta_x = \{22, \;-3, \;  7, \;  -7\} \\
    & \alpha_u = \Big \{ \left[ \begin{array}{cc} \pm 1 & 0 \end{array} \right]\t, \left[\begin{array}{cc} 0 & \pm1 \end{array}\right]\t \Big\}, \quad  
    \beta_u = \{\pm 25, \pm 25 \}.
\end{align*}
Also, two position waypoints are implemented by constraining the first two components of the state at time steps 20 and 40 of the steering horizon. 
The chance constraint linearization is performed around $\Sigma_r = 1.2 I_6$ and $ Y_r = 15 I_2$. 
All optimization problems are solved in Matlab using MOSEK~\cite{aps2019mosek} and YALMIP~\cite{Lofberg2004}.
The resulting optimal steering as well as the control effort is illustrated in Figure \ref{fig: merged}.
The feasible set in each figure is denoted with green lines and the mean of each signal with a dashed black line. 
The 3-sigma confidence level bounds are represented with blue ellipses for the state covariance and by the light-blue area around the mean control signal. 
Initial and terminal values for the state 3-sigma confidence ellipses as well as the waypoints are denoted with red.  
\begin{figure}
  \centering
  \includegraphics[width = 0.9\columnwidth]{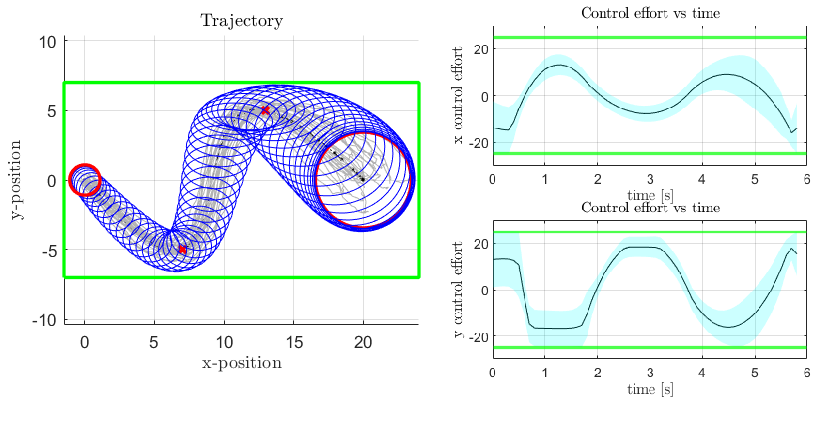}
  \caption{Left: Resulting trajectory, Right: required control effort. In both Figures, green lines represent the feasible part of the state space and control action space.}
  \label{fig: merged}
\end{figure}

In the previous example, the plant dynamics were assumed linear. In the next example, we will explore the performance of the CS algorithm for controlling the non-linear quadrotor dynamics around an aggressive reference trajectory. The nonlinear system describing a quadrotor is 
\begin{subequations} \label{NLDS}
    \begin{align}
        & \dot{r} = v, \\
        & \dot{v} = 1/m(e_3 g + R(q) \hat{e}_3 \tau + w_f),\\
        & \dot{q} = S(q) (\omega + w_m), 
    \end{align}
\end{subequations}
where $r, \; v$ represent the position and velocity of the quad in an inertial coordinate frame, $q = [\psi \; \phi \; \theta]\t$ represent the attitude, parameterized using ZYX Euler angles. The system inputs are the total thrust $\tau$ and the body-frame angular rates $\omega\t$. 
In this setup, it is assumed that the thrust and rotational velocity commands can be realized by means of some low-level controller running at a higher frequency. 
The matrix $R(q)$ is the standard rotation matrix and $S(q)$ is the matrix converting the body-frame angular rates $\omega$ to Euler angle rates $\dot{q}$. 
The total quadrotor mass is denoted by $m$, the acceleration due to gravity by $g$, and the unit vector in the $z$ direction by $e_3$. 
The vectors  $w_f, w_m$ represent disturbances. This system is known to be differentially flat, with a flat output $\zeta = [ r_x \; r_y \; r_z \;  \psi]\t$ \cite{mellinger2011minimum}.
That is, given any smooth trajectory in the flat output, the commands required to realize this trajectory as well as the values of all the state variables throughout can be described in terms of the flat output and its derivatives. 
Exploiting this result, and considering $\psi = 0$ throughout the trajectory for simplicity, we generate a smooth, discrete-time nominal path for the flat output using the model of a triple integrator.

For controlling the uncertainty, we consider a linearization of the system around the nominal trajectory. To this end, note that the nonlinear plant \eqref{NLDS} can be written in the general form $ \dot{x} = f(x, u, w)$ and discretized using a first-order difference approximation, yielding 
\begin{equation}
    x_{k+1} = F(x_k, u_k, w_k) = x_k + \Delta T f(x_k, u_k, w_k),
\end{equation}
where $\Delta T$ is the sampling step. From there, using a Taylor series expansion, the system can be approximated by 
\begin{equation*}
    x_{k+1} \approx F(\bar{x}_k, \bar{u}_k, 0) + \frac{\partial F}{\partial x} (x_k-\bar{x}_k) + \frac{\partial F}{\partial u}(u_k-\bar{u}_k)+ \frac{\partial F}{\partial w}w_k 
\end{equation*}
Noticing that $\bar{x}_{k+1} = F(\bar{x}_k, \bar{u}_k, 0)$, the deviation from the nominal trajectory $\tilde{x}_k = x_k - \bar{x}_k$ can be propagated using the LTV system
\begin{equation} \label{error_dynamics}
    \tilde{x}_{k+1} = A_k \tilde{x}_k + B_k \tilde{u}_k + (D_k + \tilde{D}) w_k,
\end{equation}
%
 %
%
%
%
where $ A_k, \; B_k, \; D_k$ are the Jacobian matrices of the Taylor expansion evaluated around the nominal trajectory and control input, and zero disturbance, $\tilde{u}_k$ is the deviation from the nominal input and $w_k$ is the external disturbance, assumed here to be zero mean white noise with unitary covariance. 
Finally, $\tilde{D}$ accounts for all discretization and linearization errors. 
The control of the deviation from the nominal trajectory can therefore be cast as a covariance steering problem subject to the dynamics of \eqref{error_dynamics}. 
To this end, consider the standard CS problem with cost matrices $Q_k = 10 I_6  \; R_k = 0.1 I_4$, and initial and final terminal covariances 
\begin{align*}
    & \Sigma_i = \blkdiag( 10^{-2} I_3, 10^{-3} I_6), \\
    & \Sigma_f = \blkdiag( 5 \cdot 10^{-4} I_3, 10^{-3} I_6) 
\end{align*}
The first $3 \times 3$ blocks of the covariance matrices for all time steps after $k = 50$ are constrained to be smaller than $2\cdot10^{-3} I_3$, guaranteeing that the trajectory stays within a prescribed tube with respect to the nominal one, which is the minimum jerk path subject to the dynamics of a discrete, triple integrator with initial and terminal boundary conditions
\begin{equation}
    \bar{x}_0 = [0, \; 0.5, \; 0, \; 0_{1\times 6}]\t, \; \bar{x}_N = [0, \; -0.5, \; 0, \; 0_{1 \times 6}]\t
\end{equation}
and four position waypoints (see Figure \ref{fig: Covariance steering}). 
The sampling time was chosen as $\Delta T = 0.01 s$ and the steering horizon was $N = 500$ steps. The resulting steering is illustrated in Figure \ref{fig: Covariance steering}, while the required control effort is in Figure \ref{fig: control effort}. 
\begin{figure}
  \centering
  \includegraphics[width = 0.8\linewidth ]{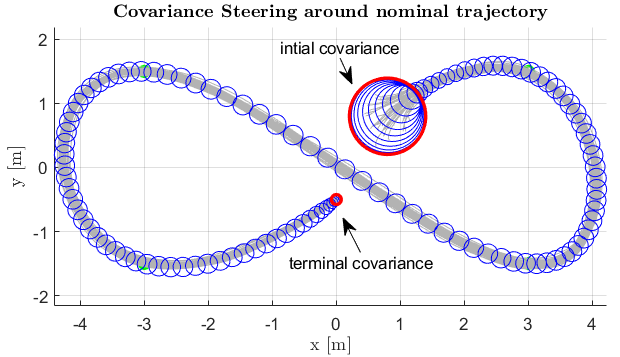}
  \caption{Uncertainty control around nominal trajectory.}
  \label{fig: Covariance steering}
\end{figure}
\begin{figure}
  \centering
  \includegraphics[width = 1 \linewidth ]{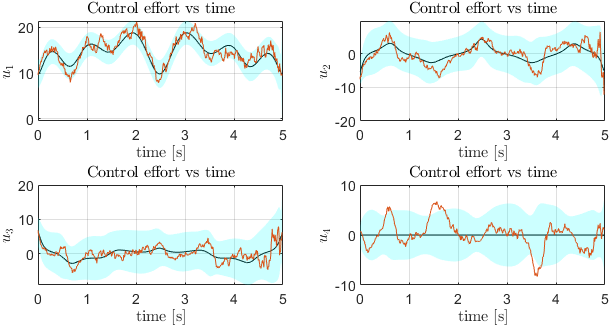}
  \caption{Required control effort}
  \label{fig: control effort}
\end{figure}

Finally, we present a run-time and resulting optimization problem size comparison between different methods for solving the unconstrained covariance steering problem. 
To evaluate the performance of each algorithm, random state space models of various sizes were generated using Matlab's \textsf{\small drss()} command. 
For each system, we use as many noise channels as state variables and half as many input channels as state variables.  
The analysis was performed for systems of varying size and a fixed steering horizon of 32 time steps, as well as for varying time horizons for an $8 \times 8$ system. 
The results are summarized in Tables \ref{ssize_analysis} and \ref{horizon_analysis} respectively. 
Run times are measured in seconds and the problem size is the number of decision variables in each program. 
The empty cells are due to the program running out of memory. 
The simulations were carried out in Matlab 2022 running on an 11\textsuperscript{th} Gen. Intel Core i7-11800H and 16 GB of RAM.
\begin{table}[ht!] 
  \centering
    \caption{Run-time comparison for varying state space size.}
  \begin{tabular}{|c|c|c|c|c|c|c|}
    \hline
    \multirow{2}{*}{\textbf{$n$}} & \multicolumn{2}{c|}{\textbf{Approach 1, \cite{bakolas2018finite}}} & \multicolumn{2}{c|}{\textbf{Approach 2, \cite{okamoto2019optimal}}} & \multicolumn{2}{c|}{\textbf{Proposed approach}}\\
    \cline{2-7}
    & \textbf{p. size} & \textbf{r. time} & \textbf{p. size} & \textbf{r. time} & \textbf{p. size} & \textbf{r. time} \\
    \hline
    4  & 3200   & 93.28 & 256   & 0.20   & 884    & 0.03   \\ \hline
    8  & 10496  & -     & 1024  & 2.91   & 3536   & 0.18   \\ \hline
    16 & 37376  & -     & 4096  & 138.07 & 14144  & 2.59   \\ \hline
    32 & 140288 & -     & 16384 & -      & 56576  & 151.76 \\ \hline
  \end{tabular}
  \label{ssize_analysis}
\end{table}
\vspace{-2 mm}
\begin{table}[h!] 
  \centering
  \caption{Run-time comparison for varying horizon size.}
  \begin{tabular}{|c|c|c|c|c|c|c|}
    \hline
    \multirow{2}{*}{\textbf{$N$}} & \multicolumn{2}{c|}{\textbf{Approach 1, \cite{bakolas2018finite}}} & \multicolumn{2}{c|}{\textbf{Approach 2, \cite{okamoto2019optimal}}} & \multicolumn{2}{c|}{\textbf{Proposed approach}}\\
    \cline{2-7}
    & \textbf{p. size} & \textbf{r. time} & \textbf{p. size} & \textbf{r. time} & \textbf{p. size} & \textbf{r. time} \\
    \hline
    8   & 640     & 3.57   & 256   & 0.12  & 848   & 0.04 \\ \hline
    16  & 2306    & 76.74  & 512   & 0.70  & 1744  & 0.08 \\ \hline
    32  & 8704    & -      & 1024  & 3.33  & 3536  & 0.17 \\ \hline
    64  & 33792   & -      & 2048  & 19.27 & 7120  & 0.36 \\ \hline
    128 & 133120  & -      & 4096  & -     & 14288 & 0.75 \\ \hline
    256 & 528384  & -      & 8129  & -     & 28624 & 1.60 \\ \hline    
  \end{tabular}
  \label{horizon_analysis}
\end{table}

It is clear that the proposed approach outperforms the state-of-the-art algorithms significantly, by over an order of magnitude for almost all cases. 
Also, it is worth noting that problem \eqref{convex_eq} is a linear semidefinite program, while the formulations of \cite{bakolas2018finite} and \cite{okamoto2019stochastic} result in quadratic semidefinite programs, which need to be converted to linear ones using suitable relaxations, increasing further the number of decision variables needed as well as the complexity of the problem. 
Finally, Problem \eqref{convex_eq} involves $N-1$ LMIs of dimensions $p \times p$ as opposed to a single large LMI of dimensions $(N+2)n \times (N+2)n$ for the terminal covariance constraint used in methods \cite{bakolas2018finite, okamoto2019stochastic}. 
As suggested in \cite{aps2019mosek}, multiple smaller LMIs can be solved more efficiently compared to a single larger one due to the resulting sparsity of the constraints. 
This also explains why although Approach~2 of~\cite{okamoto2019stochastic} results in smaller problem sizes compared to the proposed approach, still has significantly larger solution times.

\section*{ACKNOWLEDGMENT}
The authors would like to sincerely thank Dr. Fengjiao Liu for her constructive comments and discussion on the paper, and Ujjwal Gupta for his help with the quadrotor example.
Support for this work has been provided by 
ONR award N00014-18-1-2828 and NASA ULI award \#80NSSC20M0163.
This article solely reflects the opinions and conclusions of its authors and not of any NASA entity.

\bibliographystyle{ieeetr}
\bibliography{george_refs}

\begin{thebibliography}{10}

\bibitem{hotz1987covariance}
A.~Hotz and R.~E. Skelton, ``Covariance control theory,'' {\em International
  Journal of Control}, vol.~46, pp.~13--32, July 1987.

\bibitem{grigoriadis1997minimum}
K.~M. Grigoriadis and R.~E. Skelton, ``Minimum-energy covariance controllers,''
  {\em Automatica}, vol.~33, no.~4, pp.~569--578, 1997.

\bibitem{primbs2009stochastic}
J.~A. Primbs and C.~H. Sung, ``Stochastic receding horizon control of
  constrained linear systems with state and control multiplicative noise,''
  {\em IEEE Transactions on Automatic Control}, vol.~54, pp.~221--230, Feb.
  2009.

\bibitem{farina2013probabilistic}
M.~Farina, L.~Giulioni, L.~Magni, and R.~Scattolini, ``A probabilistic approach
  to model predictive control,'' in {\em 52nd IEEE Conference on Decision and
  Control}, (Firenze, Italy), pp.~7734--7739, Dec. 2013.

\bibitem{chen2015optimal}
Y.~Chen, T.~T. Georgiou, and M.~Pavon, ``Optimal steering of a linear
  stochastic system to a final probability distribution, part {I},'' {\em IEEE
  Transactions on Automatic Control}, vol.~61, pp.~1158--1169, May 2015.

\bibitem{chen2015optimal1}
Y.~Chen, T.~T. Georgiou, and M.~Pavon, ``Optimal steering of a linear
  stochastic system to a final probability distribution, part {II},'' {\em IEEE
  Transactions on Automatic Control}, vol.~61, pp.~1170--1180, May 2015.

\bibitem{bakolas2018finite}
E.~Bakolas, ``Finite-horizon covariance control for discrete-time stochastic
  linear systems subject to input constraints,'' {\em Automatica}, vol.~91,
  pp.~61--68, May 2018.

\bibitem{okamoto2018optimal}
K.~Okamoto, M.~Goldshtein, and P.~Tsiotras, ``Optimal covariance control for
  stochastic systems under chance constraints,'' {\em IEEE Control Systems
  Letters}, vol.~2, pp.~266--271, July 2018.

\bibitem{okamoto2019stochastic}
K.~Okamoto and P.~Tsiotras, ``Stochastic model predictive control for
  constrained linear systems using optimal covariance steering,'' {\em arXiv
  preprint arXiv:1905.13296}, 2019.

\bibitem{okamoto2019optimal}
K.~Okamoto and P.~Tsiotras, ``Optimal stochastic vehicle path planning using
  covariance steering,'' {\em IEEE Robotics and Automation Letters}, vol.~4,
  pp.~2276--2281, July 2019.

\bibitem{liu2022optimal_mult}
F.~Liu and P.~Tsiotras, ``Optimal covariance steering for continuous-time
  linear stochastic systems with multiplicative noise,'' {\em arXiv preprint
  arXiv:2206.11735}, 2022.

\bibitem{saravanos2021distributed}
A.~D. Saravanos, A.~Tsolovikos, E.~Bakolas, and E.~Theodorou, ``Distributed
  covariance steering with consensus {ADMM} for stochastic multi-agent
  systems,'' in {\em Proceedings of Robotics: Science and Systems}, (Virtual),
  July 2021.

\bibitem{ridderhof2019nonlinear}
J.~Ridderhof, K.~Okamoto, and P.~Tsiotras, ``Nonlinear uncertainty control with
  iterative covariance steering,'' in {\em IEEE 58th Conference on Decision and
  Control (CDC)}, (Nice, France), pp.~3484--3490, Dec. 2019.

\bibitem{saravanos2022distributed}
A.~D. Saravanos, I.~M. Balci, E.~Bakolas, and E.~A. Theodorou, ``Distributed
  model predictive covariance steering,'' {\em arXiv preprint
  arXiv:2212.00398}, 2022.

\bibitem{sivaramakrishnan2022distribution}
V.~Sivaramakrishnan, J.~Pilipovsky, M.~Oishi, and P.~Tsiotras, ``Distribution
  steering for discrete-time linear systems with general disturbances using
  characteristic functions,'' in {\em American Control Conference (ACC)},
  (Atlanta, GA, USA), pp.~4183--4190, June 2022.

\bibitem{renganathan2022distributionally}
V.~Renganathan, J.~Pilipovsky, and P.~Tsiotras, ``Distributionally robust
  covariance steering with optimal risk allocation,'' {\em arXiv preprint
  arXiv:2210.00050}, 2022.

\bibitem{liu2022optimal}
F.~Liu, G.~Rapakoulias, and P.~Tsiotras, ``Optimal covariance steering for
  discrete-time linear stochastic systems,'' {\em arXiv preprint
  arXiv:2211.00618}, 2022.

\bibitem{balci2022covariance}
I.~M. Balci and E.~Bakolas, ``Covariance steering of discrete-time linear
  systems with mixed multiplicative and additive noise,'' {\em arXiv preprint
  arXiv:2210.01743}, 2022.

\bibitem{balci2022exact}
I.~M. Balci and E.~Bakolas, ``Exact {SDP} formulation for discrete-time
  covariance steering with {Wasserstein} terminal cost,'' {\em arXiv preprint
  arXiv:2205.10740}, 2022.

\bibitem{vandenberghe1996semidefinite}
L.~Vandenberghe and S.~Boyd, ``Semidefinite programming,'' {\em SIAM Review},
  vol.~38, no.~1, pp.~49--95, 1996.

\bibitem{aps2019mosek}
M.~ApS, {\em Mosek Optimization Toolbox for {MATLAB}}, 2019.

\bibitem{Lofberg2004}
J.~L{\"{o}}fberg, ``Yalmip : A toolbox for modeling and optimization in
  {MATLAB},'' in {\em In Proceedings of the CACSD Conference}, (Taipei,
  Taiwan), 2004.

\bibitem{mellinger2011minimum}
D.~Mellinger and V.~Kumar, ``Minimum snap trajectory generation and control for
  quadrotors,'' in {\em IEEE International Conference on Robotics and
  Automation}, (Shanghai, China), pp.~2520--2525, 2011.

\end{thebibliography}
\end{document}